\documentclass[12pt,draftcls,onecolumn,english]{IEEEtran}

\usepackage[T1]{fontenc}
\usepackage[latin9]{inputenc}
\usepackage{endnotes}
\usepackage{amssymb}
\usepackage{amsmath}
\usepackage{graphicx}
\usepackage{amssymb}
\usepackage{esint}
\usepackage{cite}
\usepackage{pstricks,psfrag}
\usepackage{url,amsmath,amssymb}
 
 \let\footnote=\endnote

\def\argmax{\operatorname{arg~max}}
\def\E{\mathbb{E}}

\def\P{\mathbb{P}}
\def\p{\mathsf{P}}
\def\ie{{\em i.e.}}

\def\d{\mathrm{d}}

\def\R{\mathbb{R}}

\def\Z{\mathbb{Z}}

\def\sinr{\mathsf{SINR}}

\def\snr{\mathsf{SNR}}
\def\sir{\mathsf{SIR}}

\def\i{\mathbf{1}}

\def\x{\mathsf{x}}

\def\l{\ell}

\def\T{\theta}

  \newtheorem{lemma}{Lemma}

  \newcommand{\h}[1]{\ensuremath{\mathsf{h}_{#1}}}

  \def\b{\eta}
  \def\G{\mathsf{G}}
  \def\cl{\mu}

\begin{document}
\title{Spatial Analysis of Opportunistic Downlink Relaying in  a Two-Hop Cellular System}
\author{Radha Krishna Ganti and Martin Haenggi}
%Department of Electrical Engineering\\
%University of Notre Dame\\
%Notre Dame, IN 46556, USA\\
%\{rganti,mhaenggi\}@nd.edu}
\maketitle
\begin{abstract} 
We consider a two-hop cellular system in which the mobile nodes help the base
station by relaying information to the dead spots. 
While   two-hop cellular schemes have been analyzed previously, the 
 distribution of the node locations has not been explicitly taken into account.
 In this paper, we model the node locations of the  base stations and the 
 mobile stations as a point process on the plane and then analyze the performance
 of  two different two-hop schemes in the downlink. In one scheme the node nearest to the
 destination that has decoded  information from the base station in the
 first hop is used as the  relay. In the second scheme the node 
 with the best channel to the relay that received information in 
 the first hop acts as a relay. In both these schemes we obtain the 
 success probability of the two hop scheme,  accounting for the  interference
 from all other cells.
 We use tools from stochastic geometry and point process theory to analyze the
 two hop schemes.  Besides the results obtained a main contribution of the paper
 is to introduce a mathematical framework that can be used to analyze arbitrary
 relaying schemes. Some of the main contributions of this paper are the analytical techniques introduced for the inclusion of the spatial locations of the nodes into the mathematical analysis.
  
\end{abstract}
%% LyX 1.6.2 created this file.  For more info, see http://www.lyx.org/.
%% Do not edit unless you really know what you are doing.
\section{Introduction}
Cellular systems are the most widely deployed wireless systems and provide
reliable communication services to billions around the world.
They consist of  base stations that serve a geographical area
called  cell. In most of the present  cellular systems, the base station (BS)
communicates  directly with the mobile users (MS) in its cell.  This single-hop
 architecture makes  it is difficult for the BSs to communicate with 
 MSs at the  cell  boundary because of the distance and the
inter-cell interference. So a base station will have to increase its power to
maintain the rate of transmission.   The dead spots problem can be countered by using more
 base stations, thereby  increasing the spatial reuse. But increasing  the
 number of base stations  can be prohibitively expensive  or even impossible.  
 The problem can be addressed more effectively by  moving away from  
 the paradigm of single-hop communication and   permitting the base station  to
 communicate with mobile stations at the boundary  by using the
other intermediate MSs in its cell in a sequence of hops. Although such
multi-hopping requires some significant changes in the present cellular system
architecture, it may help to effectively combat the dead spots
problem,  and hence the cellular multi-hopping  problem is worthy to investigate
\cite{neon:2001,hung:2004}.  
In this  paper, we analyze the  benefits  of two-hop cellular communication 
by comparing its performance with a traditional single-hop cellular system. A
two-hop system,
\begin{itemize}
  \item  may provide significant benefits over single-hop communication. 
  \item  does not have the implementation complexity of larger number of hops
	(in terms of 	routing and scheduling).
\end{itemize}

When a BS transmits, multiple MSs will  be able to receive the information, and hence
these  mobile nodes can help the BS  transmit information  to the cell edge.
 Since more than one MS can act as a relay, it is not clear how to choose a subset of these
 relays in a distributed fashion so as to reduce the
interference and increase the probability of packet delivery. In this paper, we
analyze simple relay selection schemes and
compare their performance with  direct transmission. We account for 
the inter-cell interference and the spatial structure of the transmitting nodes
in the analysis.

 We  use methods
from stochastic geometry and point process theory to model and study the two-hop
cellular system.  In particular we provide techniques based on  probability generating
functional of a point process to analyze the outage probabilities, and
we  provide asymptotic results  for the outage at high $\snr$ and low BS
density. The techniques presented in this paper can be  extended to
analyze more complicated  relay selection schemes, power control mechanisms and
 other multi-hop techniques.   The major emphasis of the paper is in the
 methodology and the techniques of  the analysis  rather than the specifics of
 the communication system. For example we concentrate only on two specific  relay selection methods although many more methods have been proposed in the literature.

\subsection{Previous work}
The problem of two-hop extensions of cellular system has been studied
extensively, and a provision for  a multi-hop technique  has been included
in the A-GSM standard \cite{neon:2001,hung:2004}. In \cite{sreng2002coverage},
a MS is selected to help the BS depending on the  large-scale path-loss on the 
BS-relay link and the relay-destination link. \cite{jingmei2004performance} considers
a similar problem, but the MSs that can act as relays are assumed to be located
on a circle around the BS,  and the authors  provide various power allocation schemes and
verify their performance by simulations. 
The present problem is also
very similar to the problem of  opportunistic relay selection. In \cite{laneman2002distributed,
laneman2004cooperative} a detailed analysis of a opportunistic two-hop relaying scheme   obtaining full diversity order
using distributed space-time codes has been provided. But a distributed
space-time code requires  very tight coordination and precise signaling  among
the relays, which increases the overhead and complexity in  the system. An
alternative approach is to choose the {\em best} relay, and 
in  opportunistic  relaying (OR) \cite{bletsas2006simple} a relay is chosen so as to maximize the minimum
signal-to-noise ratio ($\snr$) of  the source-relay and the relay-destination links.  In selection
cooperation (SC)\cite{beres2006selection,michal:2008}   the  relay   with maximum
relay-destination $\snr$ is chosen and has been shown that  SC and OR provide a similar
diversity order. In  
\cite{laneman2002distributed,bletsas2006simple,beres2006selection,michal:2008},
 distributed relay selection schemes are analyzed and  asymptotes of the outage are provided for high $\snr$.  The asymptotes provided are  functions of the means of fading  coefficients  between the source, relays and the destination. 
Averaging these results with respect to the spatial distribution of the nodes is 
difficult and hence we use an alternative approach.
   In our approach we 
model the node locations in a statistical manner and incorporate this information in the analysis from the start rather than averaging over
the spatial locations at the end. Our emphasis is on low-overhead schemes that
can readily be implemented. 

The paper is organized as follows: In Section \ref{sec:sys} the system model is
introduced, assumptions stated and  the  metrics used in the paper defined. In
Section \ref{sec:direct} the outage probability in the direct connection between the
BS and its destination is  derived. In Sections  
\ref{sec:method2} and  \ref{sec:method3} the outage probability of the  two-hop
schemes employing different relay selection schemes are analyzed. The asymptotic
gain of  using the two-hop schemes over the direct connection is also studied
in these sections. In Section \ref{sec:sim} simulation results are provided and
 compared to the theory.
%%%%
  \section{System Model}
\label{sec:sys}
We assume that the BSs (cell towers) are arranged on a square lattice of density
$\lambda_b$.
\[
\Phi_{b}={\left\{ \frac{x}{\sqrt{\lambda_{b}}},\quad x\in\mathbb{Z}^{2}\right\}
}.\]
 The analysis in this paper generalizes  in a straightforward manner to any deterministic
arrangement of BS.
We assume that $n_x$ MSs are available to assist a BS $x\in \Phi_b$.
More precisely, the locations of the mobile stations   that assist the base station $x$
form a Poisson point process \cite{stoyan} (PPP) $\Phi_x$ of density
$\lambda_x(y)=\b(y-x)$. For example choosing   $\b(y) =
\mathbf{1}_y([-1/2,1/2]^2)$ and $\lambda_b =1$
would lead to a square coverage area  for each base station. We use $\i_x(A)$ to
denote the indicator function of set $A$. See Figure
\ref{fig:illus_cellular}. 
Observe that it is not necessary for a MS to be associated to its nearest BS,
\ie,  some MSs may be outside the Voronoi cell of their BS.  We further make the following
assumptions:
\begin{enumerate}
  \item The average number of  MS that each BS serves is finite, \ie,
	\[N= \int_{\R^2}\b(x) \d x <\infty .\] This assumption implies that
	$\Phi_x$  cannot be homogeneous \cite{stoyan, verejones}.
  \item The  locations of the mobile users associated with  different base stations are independent. 
\end{enumerate}
Since the number of MSs in each cell is Poisson with mean $N=\E[n_x]$, each cell is empty with probability $\exp(-N)$. We shall use $\cl$ to denote the probability that a
cell is not empty, \ie, $\cl=1-\exp(-N)$.
\begin{figure}[h]
  \begin{center}
\includegraphics[width=3.5 in]{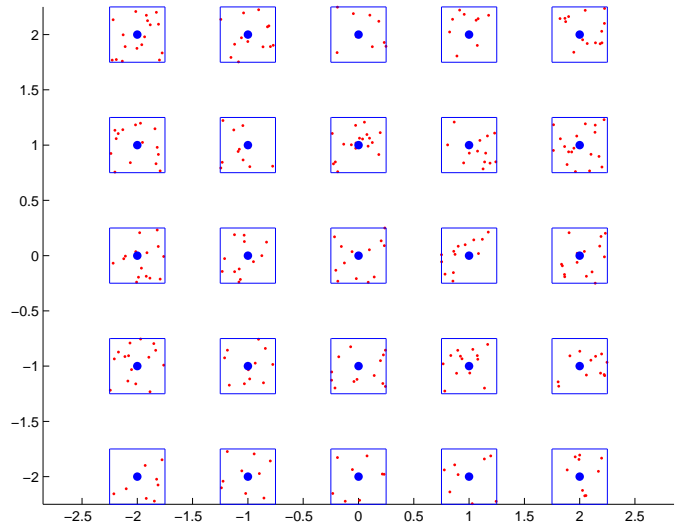}
  \end{center}
  \caption{Illustration of the cellular system with $\lambda_b=1$ and $\eta(y) =
  50\cdot\mathbf{1}_y([-0.25,0.25]^2)$. So on a average there are  $12.5$ MSs per
  cell. The bold dots represent the BSs and the smaller dots the MSs. 
  The white spaces between the cells may consist of other cells which transmit at
  a different frequency. We may model the case  where  the  neighboring
  cells  use the same frequency by choosing $\eta(y) =
  \mathbf{1}_y([-0.5,0.5]^2)$. }
  \label{fig:illus_cellular}
\end{figure}
Independent Rayleigh fading is assumed between any pair of nodes and also across
time, and the power fading coefficient between a node $x$ and node $y$ is denoted
by $\h{xy}$. Hence $\h{xy}$ is an exponential random
variable with unit mean.  The path-loss model is denoted by 
$\l(x):\ \R^{2}\setminus\{o\}\rightarrow\R^{+}$ and is a continuous, positive,
non-increasing function of $\Vert x\Vert$  that satisfies  \begin{equation}
\int_{\R^{2}\setminus B(o,\epsilon)}\l(x)\d x<\infty,\quad
\forall\epsilon>0,
\label{eq:int}
\end{equation}
where $B(a,r)$ denotes a disc of radius $r$ centered around $a$.
$\l(x)$ is usually taken to be a power law in one of the  forms:
\begin{enumerate}
  \item Singular path-loss model: $\Vert x\Vert^{-\alpha}$.
  \item Non-singular path-loss model: $(1+\Vert x\Vert^{\alpha})^{-1}$ or
	$\min\{1,\Vert x\Vert^{-\alpha}\}$.
\end{enumerate}
The integrability condition \eqref{eq:int} requires $\alpha>2$ in all the
above models.
Assuming simple linear receivers and treating interference as noise, the communication between $x$ and $y$ is successful if 
\begin{equation}
  \sinr(x,y,\Phi) = \frac{p_x\h{xy}\l(x-y)}{\sigma^2+\sum_{z\in \Phi}p_z\h{zy}
  \l(z-y)} >\T.
  \label{SINR}
\end{equation}
We also assume $\T >1$ which implies at most one transmitter can connect to a
receiver.
Here $\Phi$ is the set of interfering transmitters, $p_z$ is the
 transmission power used by a transmitter located at $z$ and  $\sigma^2$ is the 
 the additive white Gaussian noise power at the receiver. We make the following
 assumptions:
 \begin{enumerate}
   \item In the two-hop schemes that will be analyzed,  BSs
	 transmit in the even time slots and the MSs transmit in the odd time slots,
	 synchronized 	 across all cells.
   \item Each  base station $x$   has an additional mobile station, the {\em
	 destination} at $r(x)$ with $\|r(x)-x\|=R$, to which
the BS wants to transmit information.  {This additional
node just receives and never transmits.} 
\item All the BSs transmit with equal  power $P$.
 \end{enumerate}

\emph{Notation: }
\begin{itemize}
\item Define \[
\mathbf{1}(x\rightarrow y\mid\Phi)=\mathbf{1}(\sinr(x,y,\Phi)>\T).\]
 $\mathbf{1}(x\rightarrow y\mid\Phi)$ is the indicator  random variable
that is equal to one if a transmitter at $x$ is able to connect to
a receiver $y$ when the interfering set is $\Phi$.
\item Define \[
\hat{\Phi}(x)=\left\{ y\in\Phi_{x}:\ \mathbf{1}(x\rightarrow
y\mid\Phi_{b}\setminus\{x\})\right\} .\]
 $\hat{\Phi}(x)$ is  the set of MSs in the  cell  of BS $x $ to which
the BS  $x$ is able to connect in the first hop (even time slots). 
\end{itemize}

\emph{Metric:} Let $\p_{d}$ denote the probability that a BS
can connect to its destination directly in the first hop. Since all BSs are
identical \begin{equation}
  \p_{d}=\mathbb{E}\mathbf{1}(o\rightarrow r(o)\mid \Phi_b\setminus\{o\}).\label{eq:p1}\end{equation}
where $o$ denotes the origin $(0,0)$.
 A BS can connect to multiple MSs in its cell, and these connected MS are the  potential transmitters in the second hop. In the relay
selection methods studied in the next section, a subset of these potential
transmitters $R_{x}\subseteq\hat{\Phi}(x)$ are selected for each
$x\in\Phi_{b}$ to transmit in the next hop. Let the probability that
a relay can connect to its intended destination (determined by the
source to which it connects in the first hop) in the second hop be
$\p_{r}$, \ie, \begin{eqnarray}
  &  & \p_{r}=1- \mathbb{E} \prod_{\x\in R_o}1-\mathbf{1}(\x\rightarrow r(o)\mid
  \Psi\setminus\{\x\}),\label{eq:p2}\end{eqnarray}
where $\Psi=\bigcup_{x\in\Phi_{B}}R_{x}$ is the set of all transmitters in the second hop. Here we are assuming
no cooperative communication between nodes which have the same information, 
and hence relays belonging to the same cluster also
interfere with each other in the second hop.  Since $\T>1$, at most one
transmitter can connect to a receiver and thus  
\begin{eqnarray}
  &  & \p_{r}= \mathbb{E} \sum_{\x\in R_o}\mathbf{1}(\x\rightarrow r(o)\mid
  \Psi\setminus\{\x\}),\label{eq:p3}\end{eqnarray}
and the probability of success for the two-hop scheme is 
\[\p_s = 1-(1-\p_d)(1-\p_r).\]
The BS can potentially transmit in the second hop instead of using the MS as
intermediate relays. This  retransmission scheme will be used as the   base
reference,
and  the performance  of the relay selection schemes will be compared with this
retransmission scheme.
The gain in using the two-hop scheme over the retransmission scheme can be
characterized as
 \begin{equation}
   \G(\snr,\lambda_b) =
   \frac{(1-\p_d)^2}{(1-\p_d)(1-\p_r)}=\frac{1-\p_d}{1-\p_r},
   \label{eq:cel:gain}
 \end{equation}
 where \[\snr=\frac{P\l(R)}{\sigma^2}\] is the  received SNR for the direct
 transmission. To compare the direct transmission with the relay selection
 scheme,  power is allocated   across the selected relays in the second hop
 so that  the total power is equal to $P$. Another  pertinent metric to capture 
	 the performance of the network is the diversity gain, defined as
\[\d(\lambda_b) = -\lim_{\snr \rightarrow \infty}\frac{\log(1-\p_s)}{\log(\snr)}.\]
From the definition of the diversity and the gain, the following relation
follows:
\[\d_2(\lambda_b)- \d_d(\lambda_b)=\lim_{\snr \rightarrow
\infty}\frac{\log(\G(\snr,\lambda_b))}{\log(\snr)},\]
where $\d_d$ is the diversity gain for the  single-hop retransmission scheme, and
$\d_2$ is the diversity gain of the two-hop scheme. { From the definition of
$\p_s$ it can be observed that the information received in the two time slots is
decoded independently.} 

In the next sections, we will analyze the success probability  $\p_r$ and the
diversity order of the relay selection schemes. It is easy to observe that
the probability $\p_r$ of any  relay selection scheme does not tend to one by
increasing the   $\snr$  because of the interference caused by 
transmissions in other cells. So  to evaluate the asymptotic performance
of the system, we scale the BS density as
\begin{equation}
  \lambda_b = \snr^{-\beta}, \quad \beta \geq 0.
  \label{eq:scale1}
\end{equation}
As will be evident in the next section, if the signal-to-interference ratio is
defined as
\begin{equation}\sir = \frac{\l(R)}{\sum_{\x \in \Phi_b\setminus\{ o\}}\l(\x-r(o))},
  \label{eq:sir}
\end{equation}
 the scaling in \eqref{eq:scale1} translates to
\[\sir =\Theta(\snr^{\frac{\alpha\beta}{2}}).\]
So the system is interference-limited when $\beta < 2/\alpha$
and noise-limited otherwise. Hence the scaling in \eqref{eq:scale1} helps us
 evaluate the performance of the system  by varying $\beta$.
In practice this scaling can be achieved by frequency planning and decreasing the spatial
resuse factor.  We now begin with the analysis of the  direct transmission  scheme. 
%%%%%%%%%%%%%%%%%%%%
\section{First Hop: Base Station Transmits}
\label{sec:direct}
\subsection{Direct Connection}
When the BSs transmit, the inter-cell interference, fading and the noise may
cause the transmission to fail. The probability of direct connection is given 
by 
\begin{eqnarray}
  \p_{d}&=&\mathbb{E}\mathbf{1}(o\rightarrow r(o)\mid \Phi_b \setminus \{o\})\\
  &=&\mathbb{P}\left(\frac{\h{xy}\l(R)}{\frac{\sigma^2}{P}+\sum_{y\in
  \Phi_B\setminus \{o\}}\h{yr(o)}
  \l(y-r(o))} >\T \right)\nonumber \\
  &=&\exp\left(-\frac{\T\sigma^2}{P\l(R)}\right)\prod_{y\in \Phi_b\setminus
  \{o\}}\frac{1}{1+\frac{\T}{\l(R)}\l(y-r(o))}\nonumber \\
 &=&\exp\left(-\frac{\T}{\snr}\right)\Delta(r(o)),  \label{eq:direct}
\end{eqnarray}
where  \[\Delta(x)= \prod_{y\in \Phi_b\setminus
  \{o\}}\frac{1}{1+\frac{\T}{\l(x)}\l(y-x)}.\]
%  \begin{figure}
%  \begin{center}
%\includegraphics[width=3.5 in]{figures/cellular/first_hop_success.eps}
%  \end{center}
%  \caption{$\p_d$ versus $\|R\|$ for $\l(x)=\|x\|^{-4}$, $\T=2$.}
%  \label{fig:illus_first_hop_success}
%\end{figure}
%\begin{figure}
%  \begin{center}
%\includegraphics[width=3.5 in]{figures/cellular/first_hop_success_non_sing.eps}
%  \end{center}
%  \caption{$\p_d$ versus $\|R\|$ for $\l(x)=1/(1+\|x\|^{4})$, $\T=2$.}
%  \label{fig:illus_first_hop_success_non_sing}
%\end{figure}
%From Figure \ref{fig:illus_first_hop_success}, we observe  that the noise  
%decides the outage for nodes near to the BS while interference dictates
%the outage for far-away nodes when $\l(x)=\|x\|^{-4}$, \ie, a  singular
%path-loss model. This is because for the singular path-loss model the received
%power  is very large for nodes close to the  BS and the only limiting factor
%is the noise. 
%From Figure \ref{fig:illus_first_hop_success} we see that interference is a 
%big limiting factor when a non-singular path-loss model is used.
The following lemma is required  to analyze the asymptotics of the success
probability.
  \begin{lemma}
	When $\l(x)=\|x\|^{-\alpha}$ or $\l(x)=1/(1+\|x\|^{\alpha})$, 
	\[\lim_{\lambda_b\rightarrow0}\frac{1-\Delta(x)}{\lambda_b^{\alpha/2}}=\frac{\T
	C(\alpha)}{\l(x)},\]
	where 
	\begin{equation}
	  C(\alpha) =
	  \frac{\xi(\alpha/2,0)\left[\xi(\alpha/2,1/4)-\xi(\alpha/2,3/4)\right]}{2^{\alpha-2}}.
	  \label{eq:epstien}
	\end{equation}
   $\xi(s,b)= \sum_{k=0, k\neq -b}^\infty (k+b)^{-s}$ is the generalized Riemann zeta
   function. 
   \label{lemma1}
 \end{lemma}
  \begin{proof}
	We  consider the case of $\l(x)=\|x\|^{-\alpha}$;  the other case follows
	similarly. 	From the 	 definition of $\Delta(x)$  it follows that
	\begin{eqnarray*}
	  \exp\Big(-\T\l(x)^{-1}\sum_{y\in\Phi_b\setminus\{o\}}\l(y-x)\Big)
	  &\leq& \Delta(x)\leq  
	  \Big(1+ \T\l(x)^{-1}\sum_{y\in\Phi_b\setminus\{o\}}\l(y-x)\Big)^{-1}.
	\end{eqnarray*}
    We have	
	\begin{eqnarray*}
	  \sum_{y\in\Phi_b\setminus\{o\}}\l(y-x)&=&\sum_{y\in\Z^2\setminus\{o\}}\l\left(\frac{y}{\sqrt{\lambda_b}}-x
	  \right)\\
	  &=&\lambda_b^{\alpha/2}\sum_{y\in\Z^2\setminus\{o\}}\l(y-x\sqrt{\lambda_b}).
	  \label{}
	\end{eqnarray*}
	Dividing both sides by $\lambda^{\alpha/2}$ and taking the limit, the result  follows from the definition of the Epstein zeta function
\cite{epstein}.
  \end{proof}
We have $C(3)\approx9.03362$ and  $C(4)\approx 6.02681$. From the derivation of
the above lemma we observe that 
$\sir \sim \snr^{\alpha\beta/2}\l(R)C(\alpha)^{-1}$ where the definition of $\sir$
is provided in \eqref{eq:sir}.
  Using the above lemma, the asymptotic expansion of $\p_d$ for
  $\lambda_b=\snr^{-\beta}, \beta\neq 0$, at high
  $\snr$ is
  \begin{equation}
	\p_d \sim \left\{ \begin{array}{ll}
	   1-\T\snr^{-1}& \alpha\beta>2\\
	   1-\T\left( 1+C(\alpha)\l(R)^{-1}\right)\snr^{-1}&\alpha\beta=2\\
	  1-\T C(\alpha)\l(R)^{-1}\snr^{-\alpha\beta/2}& 0<\alpha\beta<2,\\
	  \end{array}\right.
	\label{eq:asymp_direct}
  \end{equation}
and  the diversity  gain of the direct transmission  is 
\[\d_d(\snr^{-\beta})=\min\left\{1,\frac{\beta\alpha}{2}\right\}.\]
So for the direct transmission, $\beta<2/\alpha$  corresponds to the interference-limited regime  and
$\beta>2/\alpha$
corresponds to the noise-limited regime. From Figure \ref{fig:direct_error} we
 note that the asymptotes in \eqref{eq:asymp_direct} are close to the true
$1-\p_r$ even at moderate $\snr$.
\begin{figure}[h]
  \begin{center}
	\includegraphics[width=12cm]{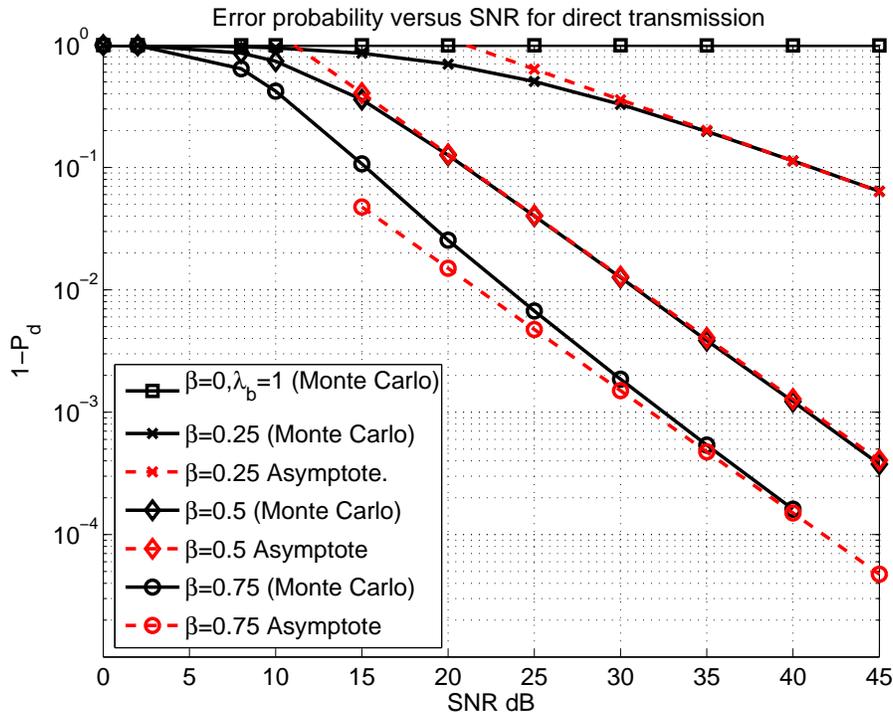}
  \end{center}
  \caption{ Outage probability $1-\p_d$ versus $\snr$  for $\lambda_b
  =\snr^{-\beta}$ with  different
  $\beta$. The system parameters are $\alpha=4$, $\T=1.5$, $r(o)=(0.5,0.5)$ and
  $\l(x)=(1+\|x\|^{4})^{-1}$. The dashed lines
  are  the asymptotes derived in \eqref{eq:asymp_direct}.  Observe the difference in the slopes
  of the error curve  for $\beta<0.5$ and $\beta\geq 0.5$.}
  \label{fig:direct_error}
\end{figure} 
In the scaling  law provided, observe that the distance of the receiver from the
BS is fixed.
%%%%
\subsection{Properties of the potential relay sets $\hat{\Phi}(x)$.}
 In this subsection, the properties of the node set that the BS at the origin is able
to connect to are analyzed. When the BSs transmit, the interference 
seen by two MSs  is independent. So the set of MSs to which  the BS at
the origin can connect to is an independent thinning of $\Phi_o$.
Hence $\hat{\Phi}(o)$ is also a PPP and since the thinning  depends on
the position, the resulting process is  inhomogeneous. Hence the intensity of
$\hat{\Phi}(o)$ is
\[\delta(x)=\b(x)\mathbb{E}\mathbf{1}(o\rightarrow x \mid\Phi_b\setminus\{o\}).\]
Following a procedure similar to the derivation of \eqref{eq:direct}, the
intensity is given by  
\begin{equation}
  \delta(x)=\b(x)\exp\left(-\frac{\T}{\snr}\frac{\l(R)}{\l(x)}\right)\Delta(x).
  \label{eq:final_density}
\end{equation}
The average number of MSs  which the BS  is able to connect to is 
\begin{eqnarray}
\mathbb{E}  \sum_{x\in \Phi_o}\mathbf{1}(o\rightarrow x
  \mid\Phi_b\setminus\{o\})  &=& \int_{\R^2}\delta(x)\d x,
  \label{eq:avg_sec}
  \end{eqnarray}
  which follows from the Campbell-Mecke theorem \cite{stoyan}. 
The average distance over which  the BS at the origin can connect is
\begin{eqnarray}
  L &=& \frac{\mathbb{E}  \sum_{x\in \Phi_o}\|x\|\mathbf{1}(o\rightarrow x
  \mid \Phi_b)}{\mathbb{E}  \sum_{x\in \Phi_o}\mathbf{1}(o\rightarrow x
  \mid \Phi_b)}\\
  &=&\frac{\int_{\R^2}\|x\|\delta(x)\d x}{\int_{\R^2}\delta(x)\d x}.
\end{eqnarray}
%If the BS does not retransmit, then there is a chance that the packet 
%might get lost in the  first slot. The gain in allowing all the MS in the
%cell to listen to the packet than just a particular receiver is equal to
%\[G_1(z)=\frac{1-P_d}{\exp(-\int_{\R^2}\delta(x)\d x)}.\] The denominator is the
%probability that the set $\hat{\Phi}(o)$ is empty and follows from the void
%probability of PPP.
%\begin{theorem}
%  When the inter cell interference is absent, \ie, $\lambda_b \rightarrow 0$ and
%  $\int\eta(x)\d x \geq 1$, a  necessary condition  for $G_1(z) >1$ is 
%  \[\int \frac{\eta(x)}{\l(x)}\d x > \frac{\int\eta(x)\d x }{\l(z)}.\]
%\end{theorem}
%\begin{proof}
%  When $\lambda_b \rightarrow 0$ the gain is equal to 
%  \begin{eqnarray}
%	G_1(z) &= & \frac{1-\exp\left(-\frac{T}{\snr}
%	\right)}{\exp\left(-\int\eta(x)\exp\left(-\frac{T}{\snr}\frac{\l(z)}{\l(x)}
%	\right)\d x\right)}
%  \end{eqnarray}
%  Since $\int\eta(x) \d x<\infty$ and $\eta(x)>0$ using Jensen's inequality 
% and the fact that $\int \eta(x)\d x \geq 1$   we have 
%    \[ \int\eta(x)\exp\left(-\frac{T}{\snr}\frac{\l(z)}{\l(x)}
%	\right)\d x > \exp\left(-\frac{T}{\snr}\int \frac{\l(z)}{\l(x)\int \eta(y)\d
%	y}\d x
%	\right)\]
%	Using the fact that $\exp(-x) < 1/x, x>0$ we have 
%	\begin{eqnarray}
%	  G_1(z) & > &
%	  \label{}
%	\end{eqnarray}
%\end{proof}
In the second hop, a subset of the MSs  which were able to receive  
information in the first-hop transmit. In the next sections we  analyze  the
following two strategies to select a subset $R_x \subset \hat{\Phi}(x)$ to transmit in the
second hop (odd time slots):
\begin{itemize}
 \item  The MS closest to the destination and that has received information in
   the first hop transmits in the second hop. This strategy requires nodes to
   know their respective locations.
 \item  The MS with the best channel (fading and path-loss) to the destination
   that has received    information in the first hop transmits. This  strategy
   requires the relays to have channel state information.
\end{itemize}

\section{Method 1: Nearest Relay to the Destination Transmits}
\label{sec:method2}
%In the previous
% selection methods, the cardinality of the set of relays  selected to  transmit
% in the second hop may be greater than one, more precisely $|R_o|$ may be
% greater than one.  This causes intra-cell  interference in the second hop and
% degrades the performance of the multi-hop  technique. In the relay selection
% method being proposed in this section, only one relay in the set $R_o$ is chosen  to
% transmit in the second hop thus reducing the intra-cell interference.  

 In this relay selection method, the node $x\in \hat{\Phi}(a), a\in\Phi_b$, closest to $r(a)$
 is selected to transmit in the second hop. To do this each node should know its
 own location, and  each packet should have location information about its
 destination.  For a fair comparison with the direct transmission scheme, we assume that  the selected
 relay transmits with power $\snr
 \sigma^2/\l(R)$. The probability of success in this relay selection method is 
 \begin{eqnarray*}
   \p_r &=&\P\left( \h{x,r(o)}\l(r) > \T (\sigma^2+I) \right),
 \end{eqnarray*}
 where  $I$ is the inter cell interference at $r(o)$, and $r$ is the distance from the  relay in the set $\hat{\Phi}(o)$  that is
 nearest to $r(o)$. More precisely
 \begin{displaymath}
   r=\left\{ \begin{array}{ll}
 \inf_{x\in \hat{\Phi}(o)} \|x-r(o)\|, & |\hat{\Phi}(o)| > 0\\
 \infty,&|\hat{\Phi}(o)|= 0.
   \end{array}
\right.
 \end{displaymath}
 $\hat{\Phi}(o)$ can be empty because of the following  two reasons:
 \begin{enumerate}
   \item The cell has no MS to begin with. The probability of this happening is
	 $1-\mu$.
   \item The BS was not able to connect to any MS in the first time slot.
 \end{enumerate}
For a fair comparison with direct transmission, we condition on  the cell at the
origin having at least   one MS
to begin with, \ie, $n_o>0$.
So
\[\p_r\mid (n_o > 0) = \p_r\cl^{-1}.\]
% $n_o>0$ does not imply $|\hat{\Phi}(o)|>0$, \ie, there are relays that the base
% station was able to connect in the first hop.  
% \[\p_r = (\p_r \mid |\hat{\Phi}(o)|>0)\p(|\hat{\Phi}(o)|>0). \]
% Since $\hat{\Phi}(o)$ is a PPP we have 
% \[\p_r = (\p_r \mid |\hat{\Phi}(o)|>0)\left(1-\exp\left(-\int_{\R^2}\delta(x)\d
% x\right)\right) .\]
 Let $F_o(r,\snr,\lambda_b)$ denote the CDF of the  first contact distribution of
 $\hat{\Phi}(o)$ from $r(o)$. It is given by 
 \begin{eqnarray}
   F_o(r,\snr,\lambda_b)&=&1-\exp\left(-\int_{B(r(o),r)}\delta(x)\d x\right).
 \end{eqnarray}
 Observe that $F_o(r,\snr,\lambda_b)$ is a defective distribution, \ie,
 $F_0(\infty,\snr,\lambda_b) <1$.  Let
 \[f_o(r,\snr,\lambda_b)= -\frac{\partial }{\partial r} F_o(r,\snr,\lambda_b)\] denote the PDF of  the first
 contact distribution. 
 Hence
 \begin{eqnarray*}
   \p_r  &=& \int_0^\infty
   \exp\left(-\frac{\T\l(R)}{\snr\l(r)}\right)
\underbrace{\E\left(\exp\left(-\frac{\T\l(R)}{\snr\sigma^2\l(r)}I\right)
   \right)}_{T_1(\lambda_b,r)}f_o(r,\snr,\lambda_b)\d r.
   \label{}
 \end{eqnarray*}
 where $I$ is the interference at $r(o)$ caused  by  transmitters in other cells. Even
 though $\int f_o(r,\snr,\lambda_b) \d r <1$ the above average is correct since the integrand
 is zero at $r=\infty$ where the remaining mass of the first contact
 distribution lies.  We now evaluate $T_1(\lambda_b,r)$.
  Let $f_a(x),  a\neq 0$, $x\in\R^2$, denote the PDF  of the nearest neighbor of
 $r(a)$ in the set $\hat{\Phi}(a)$ relative to $a$, conditioned on the event
 $|\hat{\Phi}(a)|>0$. We then have
\[ T_1(\lambda_b,r)= \prod_{a\in\Z^2\setminus\{o\}}
   \int_{\R^2}\E\left[\frac{f_a(x)}{1+\frac{\T}{\l(r)}\l(\frac{a}{\sqrt{\lambda_b}}+x-r(o)
   )\mathbf{1}(|\hat{\Phi}(a)|>0)}\right]\d x\]
  Taking the average with respect to $|\hat{\Phi}(a)|$ yields
\[\prod_{a\in\Z^2\setminus\{o\}}
   \int_{\R^2}\E\left[1-\mathbf{1}(|\hat{\Phi}(a)|>0)+\frac{f_a(x)\mathbf{1}(|\hat{\Phi}(a)|>0)}{1+\frac{\T}{\l(r)}\l(\frac{a}{\sqrt{\lambda_b}}+x-r(o)
   )}\right]\d x\]
\[\prod_{a\in\Z^2\setminus\{o\}}1-\int_{\R^2}\frac{f_a(x)(1-\exp(-\int\delta(y)\d
   y))}{1+\frac{\l(r)}{\T}\l(\frac{a}{\sqrt{\lambda_b}}+x-r(o))^{-1}}\d x.
   \]
$f_a(x)$ depends on the geometry of each cell, $\delta(x)$, $r(a)$, and  is easy to
calculate once these quantities are known. 
We now calculate the asymptotics of $\p_r$ and  the asymptotic gain.
\newline
{\em Asymptotic gain:}
In this part we scale the BS density as $\lambda_b =\snr^{-\beta}$. It is easy
to  observe that  the average number of MS in each cell that are
potential relays, \ie,  $\int \delta(x)\d x$, scales as 
\begin{equation}
  \int \delta(x)\d x \sim \int \eta(x)\d x -\frac{\T\l(R)}{\snr}\int
  \frac{\eta(x)}{\l(x)}\d x -\frac{\T C(\alpha)}{\snr^{\alpha\beta/2}}\int
  \frac{\eta(x)}{\l(x)}\d x.
  \label{eq:scale_delta}
\end{equation}
It can also be verified that  
\[\sup_{r}|F_o(r,\snr,\snr^{-\beta}) -F_o(r,\infty,0)| \rightarrow 0\]
as $\snr \rightarrow \infty$, which implies $F_o(r,\snr,\snr^{-\beta})$
converges uniformly to $F_o(r,\infty,0)$. Hence  we can interchange the
derivative and the limit  in the asymptotic analysis. 
We have 
\[f_o(r,\snr,\snr^{-\beta}) = \exp\left(-\int_{B(r(o),r)}\delta(x)\d x\right)\frac{\partial}{\partial r}
\int_{B(r(o),r)}\delta(x)\d x. \]
From \eqref{eq:scale_delta} and the fact that  $\exp(-x)\sim 1-x$ for small $x$ it follows that,
\begin{displaymath}
  f_o(r,\snr,\snr^{-\beta}) \sim \left\{ \begin{array}{ll}
	\exp\left(-f(r)\right)\left(
\frac{\partial}{\partial r}f(r)- \frac{\T\l(R)}{\snr}g(r) \right) & \alpha\beta>2\\
	\exp\left(-f(r) \right)\left(
	\frac{\partial}{\partial r}f(r) -\frac{\T C(\alpha)}{\snr^{\alpha\beta/2}}g(r) \right)&\alpha\beta<2
\end{array}\right.
\end{displaymath} 
where 
\[g(r) =\frac{\partial}{\partial r}\int_{B(r(o),r)}
\frac{\eta(x)}{\l(x)}\d x - \int_{B(r(o),r)}  \frac{\eta(x)}{\l(x)}\d x
\frac{\partial}{\partial r}f(r),\] and 
\[f(r)=\int_{B(r(o),r)}\eta(x)\d x  .\]
The following limit follows similar to the asymptotic analysis of $\p_d$
\[\lim_{\snr \rightarrow
\infty}\frac{1-T_1(\snr^{-\beta},r)}{\snr^{-\beta\alpha/2}}=\frac{\T
C(\alpha)\cl}{\l(r)},\]
   where $C(\alpha)$ is given by \eqref{eq:epstien}.
 %  We also have
 %  \begin{eqnarray*}&&\frac{\partial}{\partial
 %  r}\exp\left(-\int_{B(r(o),r)}\delta(x)\d x\right)\\
 %  &&\sim
 %  \left(1+\frac{\T\l(R)}{\snr}\gamma(r) + \T
 %  C(\alpha)\gamma(r) \snr^{-\alpha\beta/2}\right)\frac{\partial}{\partial
 %  r}\exp\left(-\int_{B(r(o),r)}\b(x)\d x\right).
 %\end{eqnarray*}
 %where $\gamma(r)= \int_{B(r(o),r)}\b(x)\l(x)^{-1}\d x$.
By some basic algebraic manipulations the asymptotic expansion of the error probability with respect to $\snr$
    with  $\lambda_b=\snr^{-\beta},\beta>0$, is
\begin{equation}
  \label{eq:asym_near}
  1-\p_2 \mid (n_o>0) \sim \left\{ \begin{array}{ll}
   \frac{\T\l(R)}{\snr \mu}\int_0^\infty
   \exp\left(-f(r)\right)\left\{g(r)+\l(r)^{-1}\frac{\partial}{\partial r}f(r)
   \right\}\d r& \alpha\beta>2 \\
   \frac{\T C(\alpha)}{\snr^{\alpha\beta/2} \mu}\int_0^\infty
   \exp\left(-f(r)\right)\left\{g(r)+\mu\l(r)^{-1}\frac{\partial}{\partial r}f(r)
   \right\}\d r& \alpha\beta<2. 
 \end{array}\right.
\end{equation}
\begin{figure}[h]
  \begin{center}
	\includegraphics[width=12cm]{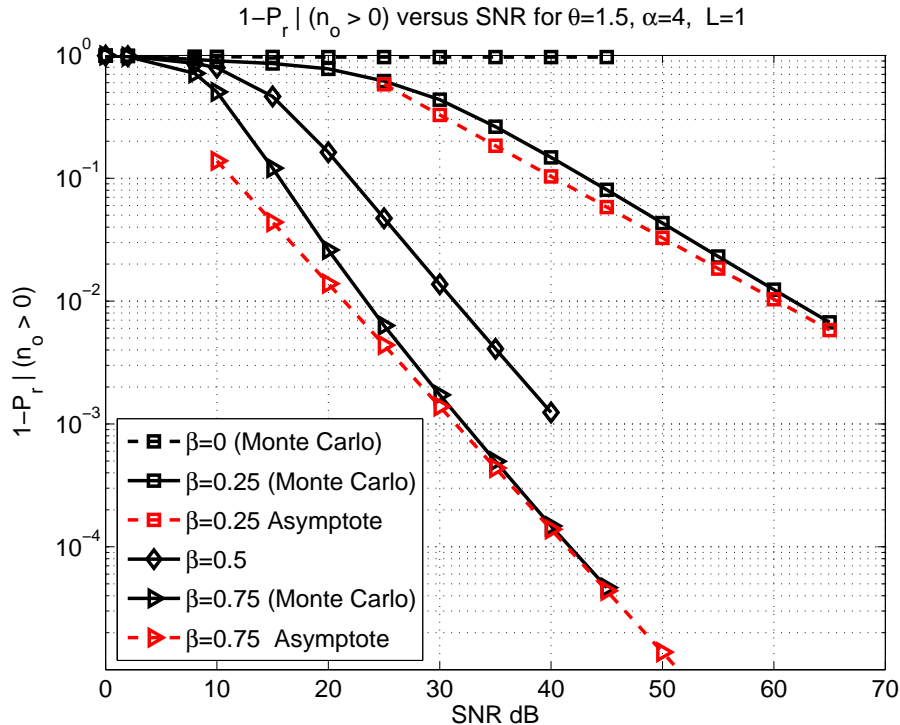}
  \end{center}
  \caption{Outage probability $1-\p_r\mid (n_o>0)$ versus $\snr$  for $\lambda_b
  =\snr^{-\beta}$ for  different
  $\beta$. The system parameters are $\alpha=4$, $\T=1.5$, $r(o)=(0.5,0.5)$, 
  $\l(x)=(1+\|x\|^{4})^{-1}$ and $\eta(y)=5\mathbf{1}_y([-0.5,0.5]^2)$. The dashed lines
  are  the asymptotes derived in \eqref{eq:asym_near}.  The dashed lines are
  the asymptotes derived in \eqref{eq:asym_near} and are approximately equal to
  $10.351 \snr^{-0.5}$ (interference-limited) and $1.387\snr^{-1}$
  (noise-limited).}
  \label{fig:near_error}
\end{figure} 
These asymptotes are plotted in Figure \ref{fig:near_error}.
From \eqref{eq:asym_near} the  asymptotic gain is 
\begin{displaymath}
  \lim_{\snr \rightarrow \infty}G(\snr,\snr^{-\beta}) \sim \left\{ \begin{array}{ll}
	\mu\l(R)^{-1}\left(\int_0^\infty
   \exp\left(-f(r)\right)\left\{g(r)+\l(r)^{-1}\frac{\partial}{\partial r}f(r)
   \right\}\d r\right)^{-1}& \alpha\beta>2 \\
   \mu \l(R)^{-1}\left(\int_0^\infty
   \exp\left(-f(r)\right)\left\{g(r)+\mu\l(r)^{-1}\frac{\partial}{\partial r}f(r)
   \right\}\d r\right)^{-1}& \alpha\beta<2. 
 \end{array}\right.
\end{displaymath}
 {\em Remarks:}
  \begin{itemize}
\item We observe that the gain is higher in the interference-limited regime than
  the noise-limited regime. This is because in this relay selection method, some
  of the cells may not be able to transmit  because they do not contain any MS,
  which  happens with probability $1-
  \cl$.
\item Since the gain does not scale with $\snr$, the diversity of this scheme is
  also equal to $\min\{1,\beta\alpha/2\}$. See Figure \ref{fig:near_error} for
  the error plot obtained by Monte Carlo
simulations  and the above asymptotes obtained theoretically.
  \end{itemize}
%\subsection{Random relay selection}
%Instead of selecting  a node in $\hat{\Phi}(o)$  that is closest to $r(o)$, we
%could have picked any node of $\hat{\Phi}(o)$ randomly. It is intuitive that
%this method  will perform worse as compared to the nearest neighbor. But this
%method does not require any geographical information and hence is easy to
%implement. The analysis of this method is similar to the previous method except
%that  the conditional nearest neighbor  distribution $f_a(r)$ is replaced by 
%the distribution of the distance of a randomly picked node from $\hat{\Phi}(a)$. More
%precisely if $y\in \R^2$ denotes a random vector  with density function
%\[\frac{\delta(r(a)-y)}{\int \delta(x) \d x},\] then $f_a(r)$ is replaced by the PDF of
%$\|y\|$. We will denote this PDF by $g_a(r)$. The rest of the analysis follows
%exactly from the previous selection method. 
% So the asymptotic expansion of the success probability with respect to $\snr$
%   and $\lambda_b=\snr^{-\beta}$ is equal to,
%   \[\p_r \mid (n_o>0)\approx 1-\T
%   C(\alpha)\gamma_r^{-1}\mu\snr^{-\alpha\beta/2}-
%   \T\l(R)\gamma_r^{-1}\snr^{-1}\]
%where 
%\[\gamma_r= \left(\int_0^\infty \l(r)^{-1}g_o(r)\d r\right)^{-1}. \]
%Hence the asymptotic gain in this case is
%equal to 
%\begin{displaymath}
%  \lim_{\snr\rightarrow \infty}G(\snr,\snr^{-\beta}) = \left\{ \begin{array}{ll}
%	{\gamma_r\l(R)^{-1}}\cl^{-1}&
%	\beta<2/\alpha\\
%	\gamma_r{\l(R)^{-1}}&
%	\beta>2/\alpha\\
%	\gamma_r \left[\frac{1+C(\alpha)\l(R)^{-1}}{\l(R)+C(\alpha)\mu
%	} \right]&\beta=2/\alpha.
%  \end{array}
%  \right.
%\end{displaymath}
%

\section{Method 2: Relay With  Best Channel to the Destination Transmits(Selection
Cooperation).}
\label{sec:method3}
In this selection procedure, the fading between a potential
relay and the destination is also included in the criterion for the relay
selection. The relay with the best channel to the destination is selected.   This method of relay selection
is called selection cooperation. In the second hop, each relay of the set
$\hat{\Phi}(o)$ can send a channel estimation packet to the destination in an
orthogonal fashion, and the destination can choose the relay with the best
channel. Alternatively, if channel reciprocity is assumed, the relays can
estimate the channel between themselves and the destination when receiving the  NACK
 and use this information to elect the best  relay  in a distributed fashion.

As in the previous section we shall find the success probability conditioned on
the cell at the origin being non-empty, \ie, $\p_r\mid n_o>0$. As indicated
earlier
\[ \p_r\mid (n_o>0) = \cl^{-1}\p_r.\] Hence we shall first calculate the
unconditional probability $\p_r$ and then multiply it with $\cl^{-1}$. The relay that is
selected is mathematically described by \[\argmax_{x\in \hat{\Phi}(o)}\{\h{xr(o)}\l(x-r(o))\}. \] The exact analysis of
this relay selection in the presence of interference is difficult and hence our
aim in this section is to  obtain the scaling behaviour of $G(\snr,\snr^{-\beta/2})$.
Let $k$  denote the cardinality of the set  $\hat{\Phi}(o)$. Since the connectivity in the first hop is
independent across relays, $k$ is a Poisson random variable with mean 
\[\E[k]=\int_{\R^2}\delta(x)\d x.\]
To make the comparison with the direct transmission easier, we assume that each node transmits with power $P=\snr\sigma^2/\l(R)$.
The probability of error is 
\[1-\p_r = \P\left( P\max_{x\in \hat{\Phi}(o)}\{\h{xr(o)}\l(x-r(o))\} <
\T(I+\sigma^2) \right),\]
where $I$ is the interference at $r(o)$ caused by concurrent transmissions in
other cells.
Conditioning on the point set $\hat{\Phi}(o)$ we have
\begin{eqnarray*}
  1-\p_r\mid \hat{\Phi}(o) 
  &=& \P\left( P\max_{x\in \hat{\Phi}(o)}\{\h{xr(o)}\l(x-r(o))\} <
\T(I+\sigma^2) \mid\hat{\Phi}(o)  \right)\\
&=&\E\left[\prod_{x\in\hat{\Phi}(o)}
1-\exp\left(-\frac{\T(\sigma^2+I)}{P\l(x-r(o))} \right)\mid\hat{\Phi}(o) \right].
\end{eqnarray*}
Since $\hat{\Phi}(o)$ is a PPP with intensity function $\delta(x)$, conditioning
on there being $k$ points in the set, each node in the set is independently distributed with
density  $\kappa(x)=\frac{\delta(x)}{\int_{\R^2}\delta(y)\d y}$.
Removing the conditioning on the locations of $\hat{\Phi}(o)$, we obtain
\begin{eqnarray}
  1-\p_r \mid (|\hat{\Phi}(o)|=k) = \E\left[1-\int_{\R^2}
  \exp\left(-\frac{\T(\sigma^2+I)}{P\l(x-r(o))} \right)\kappa(x)\d x\right]^k
  \label{eq:temp1}
\end{eqnarray}
Using binomial expansion,
\begin{equation*}
1-\p_r\mid (|\hat{\Phi}(o)|=k)=  \underbrace{\E\left[\sum_{m=0}^k (-1)^{m}{k \choose
m}\E\left[\int_{\R^2}
  \exp\left(-\frac{\T(\sigma^2+I)}{P\l(x-r(o))} \right)\kappa(x)\d x
  \right]^{m} \right]}_{a_k}.
\end{equation*}
 Hence $1-\p_r\mid
  (|\hat{\Phi}(o)|=k)$ is equal to 
  \begin{eqnarray*}
	&&\sum_{m=0}^k (-1)^{m}{k \choose
  m}\int_{\R^{2m}}\nu(x_1,\hdots,x_m)\exp\left(-\frac{\T\sigma^2\sum_{b=1}^{m} \l(x_b -r(o))^{-1}}{P}
  \right) \prod_{b=1}^{m} \kappa(x_b)d x_1\hdots \d x_m,
\end{eqnarray*}
where
\begin{eqnarray*}
 \nu(x_1,\hdots,x_m)
  &=&\E\left[\prod_{b=1}^m\exp\left(-\frac{I\T}{\l(x_b-r(o))}\right)\right]\\
  &=&\E\left[\exp\left(-I\T\varrho(x_1^m) \right)\right]\\
  &=&\E\left[\exp\left(-\T\varrho(x_1^m) \sum_{a\in\Z^2}\h{y(a)r(o)}\l(y(a)-r(o))\mathbf{1}(|\hat{\Phi}(a)|>0) \right)\right]
\end{eqnarray*}
where $y(a)$ denotes the location of the selected relay in the cell at $a$ and
$\varrho(x_1^m)=\sum_{b=1}^{m} \l(x_b
-r(o))^{-1} $.  Let $g_a(x)$ denote the PDF of $a-x$ where \[x=\argmax_{x\in
\hat{\Phi}(a)}\{\h{xr(a)}\l(x-r(a))\}.\] $g_a(x)$ is difficult to calculate and
is the reason of resorting to asymptotics. 	Since $\h{y(a)r(o)}$ is exponential
it follows that  
\[  \nu(x_1,\hdots,x_m)= \prod_{a\in
\Z^2}1-\int_{\R^2}\frac{g_a(y)(1-\exp(-\int\delta(x)\d x)}{1+\T^{-1}
\varrho(x_1^m)^{-1}\l(y+\frac{a}{\sqrt{\lambda_b}}-r(o))^{-1}}\d y.
\]
Hence the unconditional probability of error is
\[\p_r = \cl^{-1}\left[1-\exp\left(-\int_{\R^2} \delta(x)\d x\right)\sum_{k=0}^\infty a_k
\frac{(\int \delta(x)\d x)^k}{k!}.\right]\]
{\em Asymptotic gain:}
The above expansion is too unwieldy to yield any asymptotics. We shall use
\eqref{eq:temp1} to obtain the  gain in the high-$\snr$ and low-interference
regime.
Removing the conditioning in  \eqref{eq:temp1} we have
\[ 1-\p_r = \E\exp\left[-\int_{\R^2}
  \exp\left(-\frac{\T(\sigma^2+I)}{P\l(x-r(o))} \right)\delta(x)\d x\right].\]

 The above result follows from the generating function of a Poisson random
  variable. Hence the required conditional probability is
  \begin{eqnarray*}\p_r\mid
	(n_o>0) &=&\cl^{-1}\left(1-\E\exp\left[-
	\int_{\R^2}\exp\left(-\frac{\T(\sigma^2+I)}{P\l(x-r(o))} \right)\delta(x)\d  x\right]\right).\end{eqnarray*}
  An upper bound follows from Jensen's inequality: 
  \[\p_r\mid (n_o>0) \leq \cl^{-1}\left(1-\exp\left[- \int_{\R^2}
  \E\exp\left(-\frac{\T(\sigma^2+I)}{P\l(x-r(o))} \right)\delta(x)\d
  x\right]\right). \]
  Similarly a lower bound can be obtained by using the inequality
  $\exp(-x)\geq1-x$ for the inner $\exp$,
  \begin{eqnarray*}
	\p_r\mid (n_o>0)&\geq&\cl^{-1}\left( 1-\exp\left(-\int_{\R^2}\delta(x)\d
	x\right)\E\exp\left(\int_{\R^2}\frac{\T(\sigma^2+I)}{P\l(x-r(o))}\delta(x) \d x\right)\right).
\end{eqnarray*}
  To evaluate the upper and lower bounds we observe that we  will have to find
  $\E[\exp(-sI)]$.  By a procedure  similar to the derivation of
  $\nu(x_1,\hdots,x_m)$:  
 \[\E[\exp(-sI)]=\prod_{a\in
\Z^2}1-\int_{\R^2}\frac{g_a(y)(1-\exp(-\int\delta(x)\d x)}{1+s^{-1}
\l(y+\frac{a}{\sqrt{\lambda_b}}-r(o))^{-1}}\d y.\]
Recall that $\delta(x)$ is equal to 
\[\b(x)\exp\left(-\frac{\T}{\snr}\frac{\l(R)}{\l(x)}\right)
\prod_{y\in \Z^2\setminus
\{o\}}\frac{1}{1+\frac{\T}{\l(x)}\l(y/\sqrt{\lambda_b}-x)}.
\]
We now find the asymptotic lower and upper bound when $\lambda_b =\snr^{-\beta}$
 for large $\snr$.
 We first observe that 
 \[\delta(x)\sim\b(x)\left(1-\frac{\T\l(R)}{\l(x)}\snr^{-1} -
 \frac{\T}{\l(x)}C(\alpha)\snr^{-\alpha\beta/2}\right).\]
It is also easy to obtain that 
 \[\E[\exp(-sI)]\sim 1-\cl s C(\alpha).\]
 After basic algebraic manipulation, it is established that both the upper and
 the lower bounds exhibit the 
 same scaling which is
 \begin{equation}
     \p_r\mid (n_o>0) \sim \left\{
	 \begin{array}{ll}
	   1- \snr^{-1}\left(\frac{1-\cl}{\cl}\right)\T\l(R)\int_{\R^2}\left[\frac{1}{\l(x-r(o))}+\frac{1}{\l(x)}\right]\b(x)\d
   x&\alpha\beta>2\\
1-\snr^{-\alpha\beta/2}\left(\frac{1-\cl}{\cl}\right)\T C(\alpha)\int_{\R^2}\left[\frac{\cl}{\l(x-r(o))}+\frac{1}{\l(x)}\right]\b(x)\d
   x &\alpha\beta<2.
	 \end{array}\right.
   \label{eq:error_best}
 \end{equation}
Hence the gain is
\begin{equation}
  \lim_{\snr\rightarrow\infty}G(\snr,\snr^{-\beta}) =\left\{
  \begin{array}{ll}
\frac{\cl}{1-\cl}\l(R)^{-1}\left[\int_{\R^2}\left[\frac{1}{\l(x-r(o))}+\frac{1}{\l(x)}\right]\b(x)\d
   x\right]^{-1} & \alpha\beta >2 \\
\frac{\cl}{1-\cl}\l(R)^{-1}\left[\int_{\R^2}\left[\frac{\cl}{\l(x-r(o))}+\frac{1}{\l(x)}\right]\b(x)\d
   x\right]^{-1} &\alpha\beta<2
  \end{array}
  \right.
  \label{}
\end{equation}
Hence the diversity of this scheme is
\[\d_2(\snr^{-\beta})=\min\left\{1,\frac{\alpha\beta}{2}\right\}.\]
In the above analysis we assumed that the cell is non-empty and hence obtained a
maximum diversity of $1$. 
\begin{figure}[h]
  \begin{center}
	\includegraphics[width=12cm]{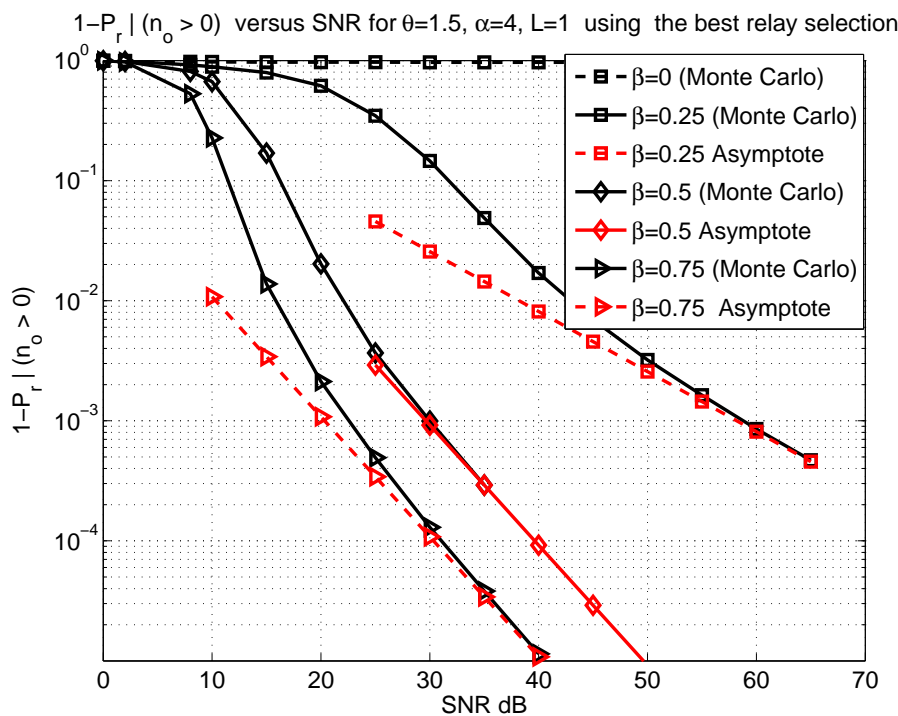}
  \end{center}
  \caption{Outage probability $1-\p_r \mid (n_o>0)$ versus $\snr$  for
  $\lambda_b =\snr^{-\beta}$ and various
  $\beta$. The system parameters are $\alpha=4$, $\T=1.5$,  $z=(0.5,0.5)$, 
  $\l(x)=(1+\|x\|^{4})^{-1}$ and $\eta(y)=5\mathbf{1}_y([-0.5,0.5]^2)$. The dashed lines
  are  the asymptotes derived in \eqref{eq:error_best} and are approximately equal to $0.812
  \snr^{-0.5}$(interference limited) and $0.108\snr^{-1}$ (noise limited).} 
  \label{fig:best_error}
\end{figure}

\section{Simulation Results and Observations}
\label{sec:sim}
In this section the gain of the proposed methods  over direct
transmission is obtained  by Monte-Carlo simulations. For the purpose of
simulation we truncate the BS lattice to  $\lambda_b^{-1/2}\{-2,-1, 0 ,1,
2\}^2$, and  $\T =1.5$ is used as the decoding threshold. The
cells are modeled as squares and the destination of each BS is located at a
random vertex of the square.  The spatial density used is 
\[\eta(y)= \lambda_m\mathbf{1}_y([-L/2,L/2]^2).\] If not specified we use
$\lambda_m=5$ and $L=1$.

%From Figure \ref{fig:all_error} we observe the effect of intra cell interference 
%when all the relays that received information in the first hop transmit. In such
%a scenario, reducing the co-cell interference would not provide additional gain.
%Correspondingly in Figure \ref{fig:gain_all},  we first observe that the gain
%$G(\snr,\snr^{-\beta})$  decreases to zero when $\beta >0$ and tends to a
%constant when $\beta=0$. When $\beta=0$  the probability of success tends to a
%constant that depends on $\lambda_b$ for  both the direct transmission, and the
%two-hop transmission and hence the gain tends to a constant limit. 
%\begin{figure}
%  \begin{center}
%\includegraphics[width=5 in]{figures/cellular/error_combined.eps}
%  \end{center}
%  \caption{Outage probability $1-\p_r\mid (n_o>0)$ versus $\snr$  for various
%  $\beta$, $L=1$.}
%  \label{fig:error_combined}
%\end{figure}
%From the first sub-figure in  Figure \ref{fig:error_combined} we observe that
%all relays transmitting performs better than the direct connection when
%$\beta=0$, which correspond to the high-interference regime. So in the scenario in
%which the adjacent cell interference cannot be adjusted (\ie, $\beta=0$), it is
%beneficial to utilize a simple two-hop scheme.
()<++>
%\begin{figure}
%  \begin{center}
%\includegraphics[width=3.5 in]{figures/cellular/gain_all.eps}
%  \end{center}
%  \caption{$G(\snr,\snr^{-\beta})$ versus $\snr$  for various $\beta$, $\eta(y) =
%  5\mathbf{1}_y([-0.5,0.5]^2)$.}
%  \label{fig:gain_all}
%\end{figure}
In Figures \ref{fig:near_error} and \ref{fig:best_error} the error probability of the
schemes employing nearest relay to the destination and the best relay are
plotted. We observe that the asymptotes obtained from theory match perfectly
with the simulation results. As predicted by theory, the diversity obtained is
$1$ when $\alpha\beta >2$ and is equal to $\alpha\beta/2$ otherwise. 
From Figure \ref{fig:gain_nearest} and \ref{fig:gain_best}, it can be seen that the gain reaches a constant
when $\snr\rightarrow \infty$. 
\begin{figure}
  \begin{center}
\includegraphics[width=3.5 in]{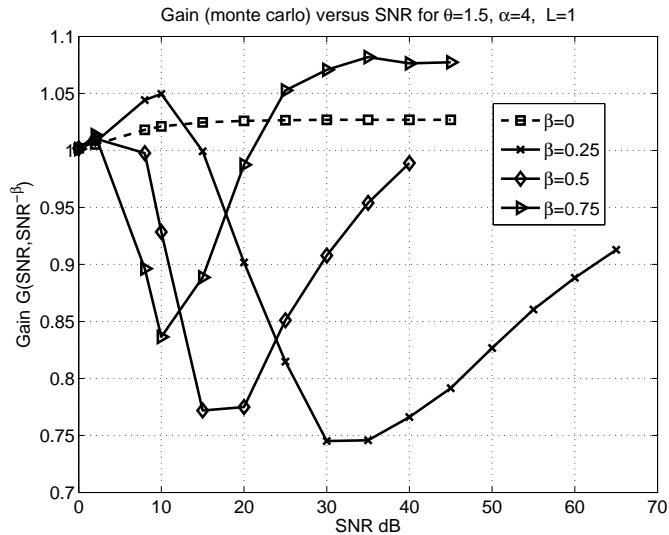}
  \end{center}
  \caption{$G(\snr,\snr^{-\beta})$ versus $\snr$  for various $\beta$. Relay
  closest to the destination is selected.}
  \label{fig:gain_nearest}
\end{figure}
\begin{figure}
  \begin{center}
\includegraphics[width=3.5 in]{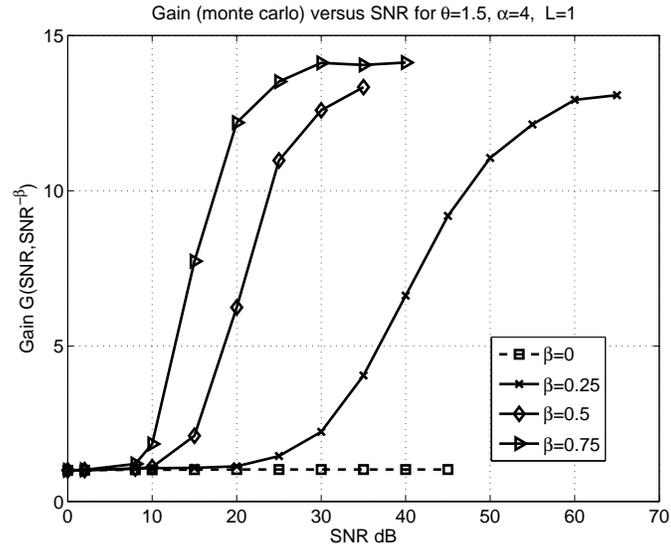}
  \end{center}
  \caption{$G(\snr,\snr^{-\beta})$ versus $\snr$  for various $\beta$. Relay
  with the best channel to the destination is selected. }
  \label{fig:gain_best}
\end{figure}
We observe that the best-relay selection scheme performs the best as expected. 
\begin{figure}
  \begin{center}
\includegraphics[width=3.5 in]{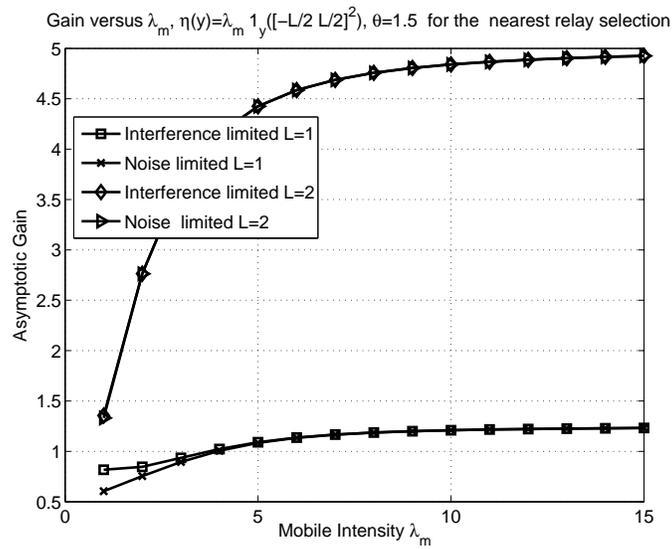}
  \end{center}
  \caption{Asymptotic gain versus $\lambda_m$ where $\lambda_m$ is the intensity
  in $\eta(y) =\lambda_m\mathbf{1}_y([-L/2,L/2]^2)$, $\l(x)=\|x\|^{-\alpha}$,
  $\T=1.5$ and $z=(-L/2,L/2)$.}
  \label{fig:near_lambda}
\end{figure}
\begin{figure}
  \begin{center}
\includegraphics[width=3.5 in]{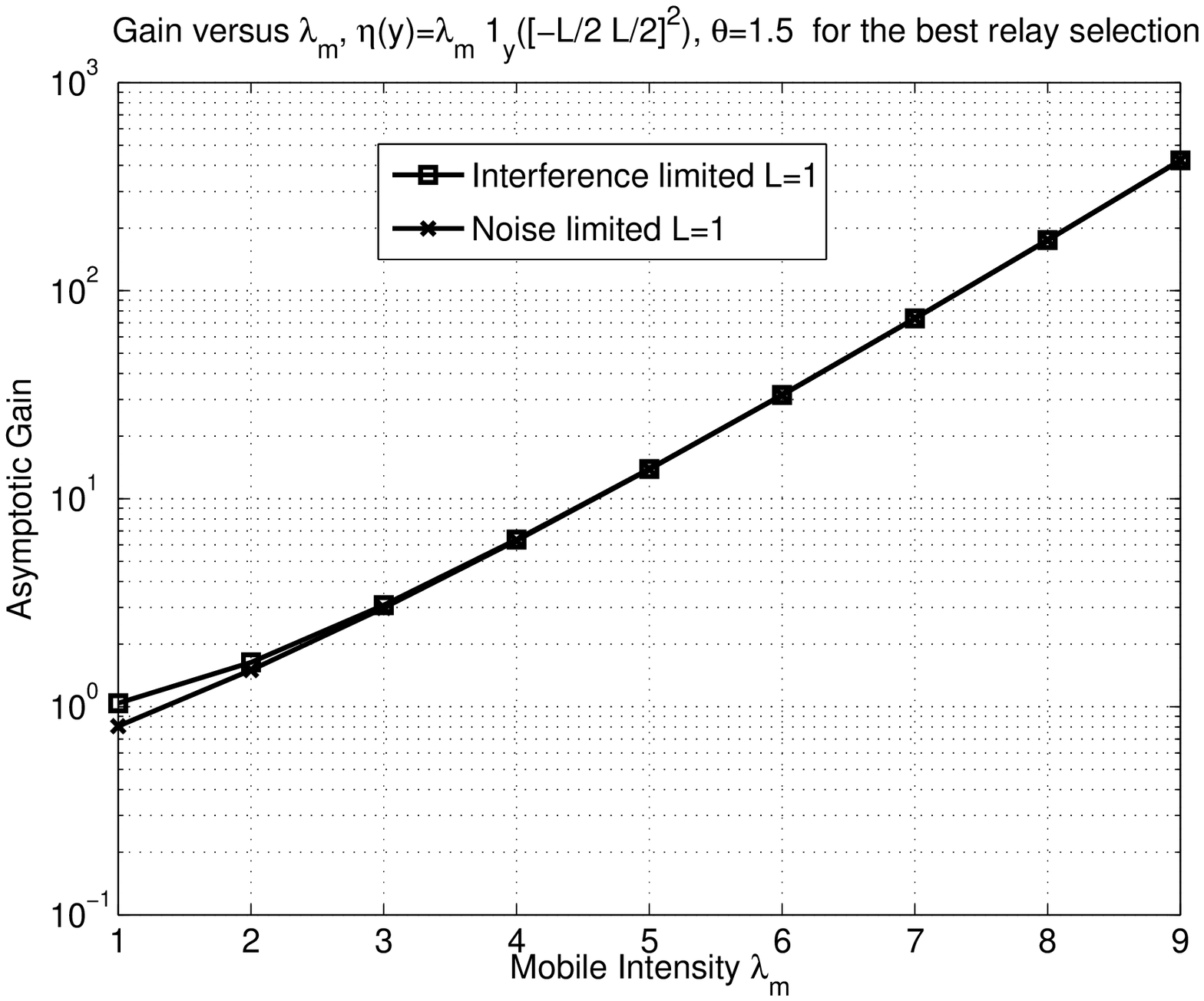}
  \end{center}
  \caption{Asymptotic gain versus $\lambda_m$ where $\lambda_m$ is the intensity
  in $\eta(y) =\lambda_m\mathbf{1}_y([-L/2,L/2]^2)$, $\l(x)=\|x\|^{-\alpha}$,
  $\T=1.5$ and $z=(-L/2,L/2)$.}

  \label{fig:best_lambda}
\end{figure}
In Figure \ref{fig:best_lambda},  we observe that the asymptotic gain increases
exponentially  with $\lambda_m$  because of  the $(1-\mu)/\mu$ factor in
the expression for the asymptotic gain.
Setting $\lambda_b= \snr^{-\beta}$ reduces the spatial reuse factor as the $\snr$
increases. The effective throughput density of the network  is equal to
$\p_2\log(1+\T)\snr^{-\beta}$ and the  maximum  of this throughput density is
the transmission capacity \cite{weber:2005}. In Figure
\ref{fig:cell_TC},  we plot $(\p_2\mid n_o>0)\log(1+\T)\snr^{-\beta}$ versus $\snr$ for various $\beta$.
We observe that for each $\snr$ there is a $\beta$ that maximizes
the throughput density, and  that as $\snr \rightarrow \infty$, the
maximizing $\beta$ tends to $0$, which is intuitive. The figure indicates  that
a throughput density of $\approx 0.1 \text{bps}/\text{m}^2$  is achieved at low
$\snr$,  and that it increases
with $\snr$.
\begin{figure}
  \begin{center}
\includegraphics[width=3.5 in]{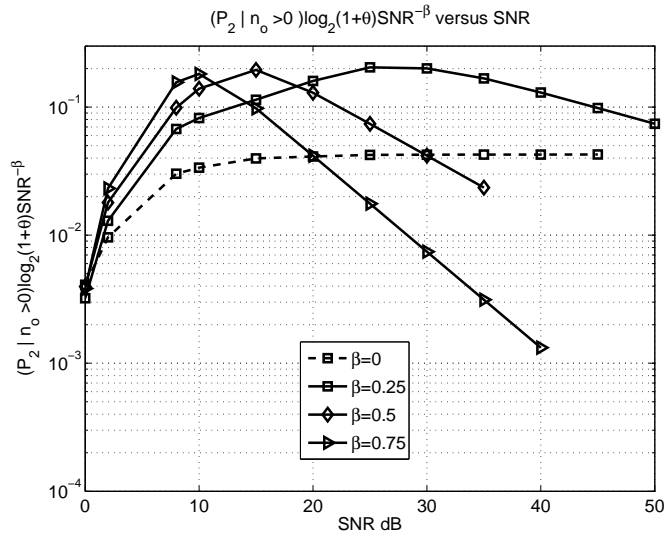}
  \end{center}
  \caption{$(\p_2\mid n_o >0)\log_2(1+\T)\snr^{-\beta}$ versus $\snr$  for various $\beta$.
  The best relay selection scheme is used.}
  \label{fig:cell_TC}
\end{figure}

\section{Conclusions}
In this paper we have analyzed the outage in a two-hop cellular system under
consideration of all the node location statistics. Outage results were provided  for two relay
selection schemes, namely  nearest-relay selection and best-relay
selection. We observed that the diversity obtained is
$\min\{1,\alpha\beta/2\}$ where $\alpha$ is the path-loss exponent, when the
density of the base stations scale as $\lambda_b =\snr^{-\beta}$ (alternatively
$\sir =\Theta(\snr^{\alpha\beta/2})$). From this
result we can infer that the system is noise-limited (even for high $\snr$) when
$\alpha\beta>2$ and interference-limited otherwise.
The asymptotic outage gain of the two-hop system over direct transmission takes  only
two values as a function of $\beta$ depending on the relay selection scheme. 
The 
gain in selecting a relay with the best channel over a direct transmission
increases exponentially with the  density of the {\em available} relays. The
gain also increases with increasing source-destination distance. From
simulations we conclude that the  gain in selecting the best relay outweighs the  overhead in estimating the fading coefficients 
between the relays and the destination as compared to the near-relay selection
method.  The techniques
 introduced in this paper can be extended for the spatial analysis of other relay
 selection schemes.

 \bibliographystyle{ieeetr}
\bibliography{point_process}	
\end{document}